\documentclass[a4paper, 11pt]{article}

\usepackage{amsmath,amsthm,amssymb,cancel,graphicx,url,verbatim,bbm}
\usepackage{tikz}
\usepackage[caption=false]{subfig}
\usepackage[margin=2cm]{geometry}

\theoremstyle{plain}
\newtheorem{theorem}{Theorem}[section]

\newtheorem{lemma}[theorem]{Lemma}

\newtheorem{question}[theorem]{Question}

\theoremstyle{definition}
\newtheorem{claim}{Claim}
\newtheorem*{claim*}{Claim}

\newcommand{\BF}[1]{{\bf\boldmath{#1}\unboldmath}}

\newcommand{\feedback}{\tau}

\newcommand{\matching}{\mu}

\newcommand{\chromatic}{\chi}

\newcommand{\cliqueCover}{\mathrm{\pi}}

\newcommand{\fractionalCliqueCover}{\mathrm{\pi^*}}

\newcommand{\IG}{\mathrm{IG}}

\newcommand{\functions}[2]{F(#1,#2)}

\newcommand{\Fix}{\mathrm{Fix}}

\newcommand{\guessing}[2]{g(#1,#2)}

\newcommand{\guessingf}{g}

\newcommand{\ShannonEntropy}{\eta}

\title{On the possible values of the entropy of undirected graphs}
\author{Maximilien Gadouleau\footnote{School of Engineering and Computing Sciences, Durham University, Durham, UK. Email: \texttt{m.r.gadouleau@durham.ac.uk}}}

\begin{document}

\maketitle

\begin{abstract}
The entropy of a digraph is a fundamental measure which relates network coding, information theory, and fixed points of finite dynamical systems. In this paper, we focus on the entropy of undirected graphs. We prove that for any integer $k$ the number of possible values of the entropy of an undirected graph up to $k$ is finite. We also determine all the possible values for the entropy of an undirected graph up to the value of four.
\end{abstract}

\section{Introduction}

\subsection{Finite Dynamical Systems and their fixed points}

\BF{Finite Dynamical Systems} (FDSs) have been used to represent a network of interacting entities as follows. A network of $n$ entities has a state $x = (x_1,\dots, x_n) \in [q]^n$, represented by a $q$-ary variable $x_v \in [q] = \{0,1,\dots,q-1\}$ on each entity $i$, which evolves according to a deterministic function $f = (f_1,\dots,f_n) : [q]^n \to [q]^n$, where $f_v : [q]^n \to [q]$ represents the update of the local state $x_v$. FDSs have been used to model gene networks \cite{Kau69, Tho73, TK01, KS08}, neural networks \cite{Hop82, Gol85, Gur97}, social interactions \cite{PS83, GT83} and more (see \cite{TD90, GM90}). 

The architecture of an FDS $f: [q]^n \to [q]^n$ can be represented via its \BF{interaction graph} $\IG(f)$, which indicates which update functions depend on which variables. More formally, $\IG(f)$ has $\{1,\dots,n\}$ as vertex set and there is an arc from $u$ to $v$ if $f_v(x)$ depends on $x_u$. In different contexts, the interaction graph is known--or at least well approximated--, while the actual update functions are not. One main problem of research on FDSs is then to predict their dynamics according to their interaction graphs. 

Among the many dynamical properties that can be studied, \BF{fixed points} are crucial because they represent stable states; for instance, in the context of gene networks, they correspond to stable patterns of gene expression at the basis of particular biological processes. As such, they are arguably the property which has been the most thoroughly studied. The study of the number of fixed points and its maximisation in particular is the subject of a stream of work, e.g. in \cite{ADG04a,RRT08,Ara08,Ric09,ARS14,GRR14}. 

\subsection{Network coding and entropy of digraphs}

\BF{Network coding} is a technique to transmit information through networks, which can significantly improve upon routing in theory \cite{ACLY00, YLCZ06}. At each intermediate node $v$, the received messages $x_{u_1}, \ldots, x_{u_k}$ are combined, and the combined message $f_v(x_{u_1},\ldots,x_{u_k})$ is then forwarded towards its destinations. The main problem is to determine which functions $f_v$ can transmit the most information. In particular, the \BF{network coding solvability problem} tries to determine whether a certain network situation, with a given set of sources, destinations, and messages, is solvable, i.e. whether all messages can be transmitted to their destinations. This problem being very difficult, different techniques have been used to tackle it, including matroids \cite{DFZ06}, Shannon and non-Shannon inequalities for the entropy function \cite{DFZ07, Rii07}, error-correcting codes \cite{GR11}, and closure operators \cite{Gad13, Gad14}. 

The network coding solvability problem can be recast in terms of fixed points of FDSs as follows \cite{Rii07, Rii07a}. The so-called \BF{$q$-guessing number} \cite{Rii07} of a digraph $D$ is the logarithm of the maximum number of fixed points over all FDSs $f$ whose interaction graph is a subgraph of $D$: $\IG(f) \subseteq D$. The guessing number is always upper bounded by the size of a minimum feedback vertex set of $D$; if equality holds, we say that $D$ is \BF{solvable} over an alphabet of size $q$ and the FDS $f$ reaching this bound is called a solution. Then, a network coding instance $N$ is solvable if and only if some digraph $D_N$ related to the instance $N$ is solvable (see \cite{GR11, Gad13, GRF15} for further illustration of the relation between network coding and the guessing number).

For a given digraph $D$, the supremum over all $q$ of the $q$-ary guessing numbers of $D$ is referred to as the \BF{entropy} of $D$. The entropy then represents the maximum amount of information which can be transmitted over a network; it is the capacity of network coding.
Although the entropy can be determined for large classes of graphs (perfect graphs, for instance), the entropy of digraphs is not so well understood in general. For instance, very little is known about the set of all possible values of the entropy.

\subsection{Outline}

The rest of the paper is organised as follows. Section \ref{sec:background} reviews some important background on the guessing number and the entropy of graphs. Section \ref{sec:results} then gives our two main results. First, we show that for any integer $k$ the number of possible values of the entropy of an undirected graph up to $k$ is finite. We then determine all the possible values for the entropy of an undirected graph up to the value of four.

\section{Background} \label{sec:background}

We consider digraphs $D = (V,E)$, where $E \subseteq V^2$; we shall usually set $V = \{v_1, \dots, v_n\}$. An undirected graph is a digraph where $E$ is a symmetric set, i.e. $(u,v) \in E$ if and only if $(v,u) \in E$; we view the pair of arcs $(u,v)$, $(v,u)$ as an edge and denote it by $uv$. A simple graph is an undirected graph without loops. For any digraph $D = (V,E)$ and any $S \subseteq V$, we denote the subgraph of $D$ induced by $S$ as $D[S]$:
$$
	D[S] = (S, \{(u,v) \in E : u, v\in S\}). 
$$ 
For any $S \subseteq V$, we also denote $D - S = D[V \setminus S]$.

We denote a bipartite graph with bipartition $A$ and $B$ and edge set $E'$ by the triple $(A,B,E')$. Then, if $G$ is a simple graph and $S, T \subseteq V$ with $T \cap S = \emptyset$, we denote the bipartite subgraph of $D$ induced by $G$ and $T$ as $G[S,T]$:
$$
	G[S,T] = (S, T, \{uv \in E : u \in S, v \in T\}). 
$$
The neighbourhood of a vertex $u$ in $G$ is $N(u;G) = \{v : uv \in E\}$; we shall usually omit the dependence in $G$. The neighbourhood of a set $S$ of vertices is $N(S) = \bigcup_{u \in S} N(u)$.

The guessing number and the entropy are formally defined as follows. Let $f: [q]^n \to [q]^n$. The \BF{interaction graph} of $f$ is the graph on $n$ vertices where $uv \in \IG(f)$ if and only if $f_v$ depends essentially on $x_u$, i.e. there exist $a,b \in A^n$ that only differ by $a_u \ne b_u$ such that $f_v(a) \ne f_v(b)$. We denote the set of all functions $f: [q]^n \to [q]^n$ interaction graph contained in a digraph $D$ as 
$$
	\functions{D}{q} = \{f : [q]^n \to [q]^n : \IG(f) \subseteq D \}.
$$
The \BF{guessing number} of $f : [q]^n \to [q]^n$ is $\guessingf(f) := \log_q |\Fix(f)|.$ The \BF{$q$-guessing number} \cite{Rii07} of $D$ is
$$
	\guessing{D}{q} := \max\{\guessingf(f) : f \in \functions{D}{q}\}.
$$
For any $D$, the $q$-guessing number tends to a limit when $q$ tends to infinity, which we shall refer to as the \BF{entropy} of $D$ \cite{Rii06, CM11}:
$$
	H(D) := \sup_{q \ge 2} \guessing{D}{q} = \lim_{q \to \infty} \guessing{D}{q}.
$$

We now review some general properties of the guessing number of digraphs found in \cite{Rii06, Rii07, CM11, GR11, GRF15}. By definition, if $D' \subseteq D$, then $\guessing{D'}{q} \le \guessing{D}{q}$. For any $D_1$ and $D_2$ on disjoint vertex sets, we have 
$$
	H(D_1 \cup D_2) = H(D_1) + H(D_2).
$$
For any digraph $D$ and any two induced subgraphs $D[V_1]$ and $D[V_2]$ of $D$ such that $V_1 \cap V_2 = \emptyset$ and $V_1 \cup V_2 = V$, we have
$$
	\guessing{D}{q} \le  \guessing{D[V_1]}{q} + |V_2|.
$$
Equality is reached in an important special case. Let $L$ be the set of vertices with a loop in $D$, then
$$
	\guessing{D}{q} = |L| + \guessing{D-L}{q}.
$$

Let $I$ be an acyclic set of $D$, then the $I$-reduction of $D$, denoted as $D^{-I} = (V^{-I}, E^{-I})$, is given as follows \cite{GRF15}. Let us denote $u \to_I v$ if there are vertices $i_1, \dots, i_k$ such that $u, i_1, \dots, i_k, v$ is a directed path in $D$. Then
\begin{align*}
	V^{-I} &= V \setminus I,\\
	E^{-I} &= \{(u,v) : u,v \in V^{-I}, (u,v) \in E \text{ or } u \to_I v\}.
\end{align*}
By \cite{GRF15}, for any acyclic set $I$ of $D$ and any $q \ge 2$,
$$
	\guessing{D}{q} \le \guessing{D^{-I}}{q}.
$$

A \BF{clique cover} of $D$ is a set of cliques $H_1, \dots, H_t$ such that every vertex belongs to at least one clique. The \BF{clique cover number} of $D$, denoted by $\cliqueCover(D)$, is the minimum number of cliques in a clique cover of $D$. Since $\guessing{K_n}{q} = n-1$ for all $n \ge 1$ and $q \ge 2$, we obtain that for any $D$ and any $q \ge 2$, 
\begin{equation} \label{eq:clique_cover}
	\guessing{D}{q} \ge n - \cliqueCover(D) \ge \matching(D),
\end{equation}
where $\matching(D)$ is the maximum matching number of $D$. In particular, if $G$ is simple, then $\guessing{G}{q} \ge n - \chromatic(\bar{G})$ for all $q \ge 2$. More specifically, if $G$ is a perfect graph, then $\guessing{G}{q} = \feedback(G) = n - \cliqueCover(G)$ for all $q \ge 2$. 

The lower bound on the guessing number in \eqref{eq:clique_cover} is refined as follows. A \BF{fractional clique cover} of a digraph $D$ is a family of cliques $H_1, \dots, H_t$ of $D$ together with non-negative weights $w_1, \dots, w_t$ such that
$$
	\sum_{i : v \in H_i} w_i \ge 1 \quad \forall\, v \in V.
$$
A clique cover is a fractional clique cover where all the weights belong to $\{0,1\}$. The minimum value of $\sum_{i=1}^t w_i$ over all clique covers is the \BF{fractional clique cover number} and is denoted by $\fractionalCliqueCover(D)$. Then there exists $k \ge 1$ such that for any $q \ge 2$ 
$$
	\guessing{D}{q^k} \ge n - \fractionalCliqueCover(D).
$$

Let $\feedback(D)$ denote the \BF{transversal number} of $D$, i.e. the minimum number of vertices which have to be removed from $D$ in order to obtain an acyclic digraph. If $G$ is a simple graph, then $\feedback(G)$ is the minimum vertex cover number of $G$. Then  for any $q \ge 2$,
$$
	\guessing{D}{q} \le \feedback(D).
$$

A finer upper bound on the entropy of a digraph is based on the submodularity of the entropy of random variables \cite{Rii07}. The \BF{Shannon entropy} of $D$ is defined as 
$$
	\ShannonEntropy(D) := \sup h(V),
$$
where the supremum is taken over all functions $h: 2^V \to \mathbb{R}$ satisfying
\begin{align*}
	h(v) &\le 1 && \forall v \in V,\\
	h(S) &\le h(T) && \forall S \subseteq T \subseteq V,\\
	h(S \cup T) + h(S \cap T) &\le h(S) + h(T) && \forall S,T \subseteq V,\\ 
	h(N(v) \cup \{v\}) &= h(N(v)) && \forall v \in V.
\end{align*}
We then have for any $q \ge 2$
$$
	\guessing{D}{q} \le \ShannonEntropy(D).
$$
Based on the Shannon entropy, we obtain that for odd cycles, $H(C_{2l+1}) = l + 1/2$ and for their complements, $H(\overline{C_{2l+1}}) = 2l - 1 - 1/l$ for any $l \ge 2$ \cite{CM11}.

\section{Main results} \label{sec:results}

Let the set of all entropies of simple graphs be
$$
	\mathcal{H} = \{H(G) : G \text{ simple, finite}\}.
$$
We remark that this is the same as considering undirected graphs.


By the results above, the set $\mathcal{H}$ is closed under addition and contains all natural numbers. The only non-integer values known so far are based on odd cycles $H(C_{2l+1}) = l + 1/2$  and their complements $H(\overline{C_{2l+1}}) = 2l - 1 - 1/l$ for $l \ge 2$ \cite{CM11}. Our first main result is that $\mathcal{H} \cap [0, k]$ is finite for every integer $k$; more precisely, it is $O(2^{6k^2})$.

\begin{theorem} \label{th:entropy_finite}
For any positive integer $k$, 
$$
	|\mathcal{H} \cap [k-1, k]| \le \sum_{h=1}^k \sum_{m = \lfloor h/2 \rfloor}^{h-1} 2^{2m(3m-2)}.
$$
\end{theorem}

\begin{lemma} \label{lem:bipartite}
Let $G = (A,B,E)$ be a nonempty bipartite graph with $|A| \ge |B| = n$. Then there exists a nonempty $A' \subseteq A$ such that $G[A', N(A')]$ has a matching saturating $N(A')$.
\end{lemma}

\begin{proof}
This is clearly true for $n = 1$. Let $G$ be a counterexample with minimal $n \ge 2$. Suppose there exists a nonempty $A' \subset A$ such that $|N(A')| \le |A'| < n$, and denote $D = G[A'' \cup N(A'')] = G[A'', N(A'')]$. We remark that for any $A'' \subseteq A'$, we have $N(A''; D) = N(A'; D)$. By minimality hypothesis there exists $A'' \subseteq A'$ such that $D[A'', N(A''; D)] = G[A'', N(A''; G)]$ has a matching saturating $N(A''; D) = N(A'';G)$. Hence for all $A' \subseteq A$, $|N(A')| \ge |A'|$ and by Hall's marriage theorem there exists a matching in $G$ saturating $A$. Thus $|A| = n$ and we obtain a matching in $G[A, N(A)] = G$ saturating $B = N(A)$.
\end{proof}

We say $G$ is \BF{entropy-minimal} if for any $G'$ with $|G'| < |G|$, $H(G') - \lfloor H(G') \rfloor \ne H(G)  - \lfloor H(G) \rfloor$. For any nonempty $S \subseteq V$, we denote 
$$
	c(S) := \{v \in V \setminus S : N(v) \subseteq S\}.
$$
We remark that $c(S)$ is an independent set.

\begin{lemma} \label{lem:claim}
Let $G$ be entropy-minimal. Then for any nonempty $S \subseteq V$, $G[c(S), S]$ does not have a matching saturating $S$.
\end{lemma}

\begin{proof}
Suppose there exists $S$ such that $G[c(S), S]$ has a matching saturating $S$. Denoting $d(S) = c(S) \cup S$, we prove that $H(G) = |S| + H(G - d(S))$, which contradicts entropy-minimality of $G$. By reducing $c(S)$, we obtain $H(G) \le H(G^{-c(S)})$. Since $G^{-c(S)}$ has a loop on every vertex in $S$ and $G^{-c(S)}- S = G - d(S)$, we obtain 
$$
	H(G) \le H(G^{-c(S)}) = |S| + H(G^{-c(S)}- S) = |S| + H(G - d(S)).
$$ 
For the reverse inequality, let $f \in \functions{G}{q}$ achieve the guessing number for $G - d(S)$. We extend it to $f' \in \functions{G}{q}$ by considering the matching in $G[c(S), S]$ saturating $S$. Let $c_1s_1, \dots, c_ks_k$ be the matching saturating $S$ where $c_i \in c(S)$ and $s_i \in S$, then
$$
	f'_v(x) = \begin{cases}
		f_v(x_{V \setminus d(S)}) &\text{if } v \in V \setminus d(S)\\
		x_{c_i} &\text{if } v = s_i \quad 1 \le i \le k\\
		x_{s_i} &\text{if } v = c_i \quad 1 \le i \le k\\
		0 &\text{otherwise.}
	\end{cases}
$$
Then $\guessingf(f') = \guessing(f) + k$ and we obtain $H(G) \ge |S| + H(G - d(S))$. 
\end{proof}

\begin{proof}[Proof of Theorem \ref{th:entropy_finite}]
Let $G$ be entropy-minimal with $H(G) \in (h-1, h)$ for some $h \le k$. Let $M$ be the vertices in a maximum matching of $G$. We then have $h \le 2m = |M| \le 2(h-1)$. The induced subgraph $G - M$ is an independent set, hence $c(M) = V \setminus M$ and the edges in $G$ either belong to $G[M]$ or to $G[c(M), M]$. Moreover, for every $c' \subseteq c(M)$, $N(c'; G) = N(c'; G[c(M), M])$. If $|c(M)| \ge |M|$, then by Lemma \ref{lem:bipartite}, there exists $c' \subseteq c(M)$ such that $G[c', N(c'; G)]$ has a matching saturating $N(c';G)$ and by Lemma \ref{lem:claim}, $G$ is not entropy-minimal. Therefore, $|c(M)| < 2m$. 

Altogether, there are at most $2^{2m(m-1)}$ choices for $G[M]$, at most $2^{2m(2m-1)}$ choices for $G[c(M), M]$ and hence at most $2^{2m(3m-2)}$ choices for $G$. It is clear that choosing $G$ empty, apart from the matching in $M$, yields an entropy of $h-1$. Also, choosing $G$ to be empty, apart from the matching in $M$ and two more edges $cm_1, cm_2$ where $c \in c(M)$ and $m_1m_2$ is an edge of the matching yields an entropy of $h$. Thus, adding up for all possible values of $h$ and $m$, we obtain the result.
\end{proof}

In Question \ref{q:H'} below, we then ask the corresponding question for the set of entropies of digraphs. The closure solvability problem \cite{Gad13} generalises the network coding solvability problem. Then for any closure operator on $V$ we can associate its entropy; the entropy of a digraph is then the entropy of the corresponding closure operator. The analogue of Question \ref{q:H'} for the set $\mathcal{H}'' \supseteq \mathcal{H}$ of entropies of closure operators was answered negatively in \cite{Gad14}: $\mathcal{H}'' \cap [0,k]$ contains all rational numbers between $1$ and $k$.

\begin{question} \label{q:H'}
Let $\mathcal{H}' = \{H(D) : G \text{ digraph, finite}\}$. Is $\mathcal{H}' \cap [0,k]$ finite for all integers $k$?
\end{question}

We are now interested in determining the first values in $\mathcal{H}$. The first three values in $\mathcal{H}$ are already known. We have $\mathcal{H} \cap [0,1] = \{0,1\}$, since a graph is either empty (and its entropy is equal to zero) or it contains an edge (and its entropy is at least one). The entropy of any directed graph greater than one is at least equal to two \cite{Gad13}. We shall give an elementary and short proof of this result for undirected graphs (namely $\mathcal{H} \cap [1,2] = \{1,2\}$). We then determine the values in $\mathcal{H}$ up to the value four. We also show that if $G$ is a connected undirected graph with entropy $5/2$, then $G$ is the pentagon; similarly, if $G$ is a connected undirected graph with entropy $11/3$, then $G$ is the graph $G_1$ displayed on Figure \ref{fig:graphs}.

\begin{theorem} \label{th:first_values}
$\mathcal{H} \cap [0, 4] = \{0, 1, 2, 5/2, 3, 7/2, 11/3, 4\}$.
\end{theorem}

\begin{lemma} \label{lem:wheel}
Let $G$ be a graph on six vertices $V = \{v_1, \dots, v_6\}$ where $G[\{v_1, \dots, v_5\}] = C_5$. Then $H(G) = 2.5$ if $v_6$ is isolated, $H(G) = 3.5$ if $N(v_6)$ contains three consecutive vertices of $C_5$, and $H(G) = 3$ otherwise.
\end{lemma}

\begin{proof}
The result is clear when $v_6$ is isolated. Suppose that $N(v_6)$ contains three consecutive vertices, say $v_1$, $v_2$, and $v_3$. We have $H(G) \le H(C_5) + 1 = 3.5$. Conversely, by fractional clique packing, we obtain $H(G) \ge 3.5$. Otherwise, $N(v_6) \subseteq S$, where $S$ is a minimum feedback vertex set of $G[\{v_1, \dots, v_5\}]$. Therefore, $S$ is a minimum feedback vertex set of $G$ of size $3$, thus $H(G) \le 3$. Conversely, say $v_5v_6 \in E$, then $\{v_1v_2, v_3v_4, v_5v_6\}$ is a matching of size $3$, thus $H(G) \ge 3$.
\end{proof}

\begin{proof}[Proof of Theorem \ref{th:first_values}]
We consider a graph $G$ such that $H(G') \ne H(G)$ for all $|G'| < |G|$. By minimality, $G$ has no isolated vertices. By construction, we have $\{0, 1, 2, 2.5, 3, 3.5, 4\} \subseteq \mathcal{H} \cap [0, 3.5]$. We now prove that the only other value in $\mathcal{H} \cap [0, 4]$ is $11/3$.

Suppose $G$ has entropy between $1$ and $2$. Let $ab \in E$, then all the other vertices must be adjacent to $a$ (or by symmetry, to $b$), i.e. $V = N(a) \cup a$. The neighbourhood of $a$ must form an independent set, for if $b, c \in N(a)$ are adjacent, then $G$ contains the triangle $a,b,c$ and $H(G) \ge 2$. But then, $\{a\}$ is a vertex cover and $H(G) \le 1$.

Suppose $G$ has entropy between $2$ and $3$. Since $G$ is not perfect, it contains $C_5$ as an induced subgraph. If $G = C_5$, then $H(G) = 2.5$. Otherwise, Lemma \ref{lem:wheel} shows that $H(G) \ge 3$.

Suppose $G$ has entropy between $3$ and $4$. Then $G$ is entropy-minimal. If $G$ contains $C_7$ as an induced subgraph, then either $G = C_7$ and $H(G) = 3.5$ or $G$ has eight or more vertices, and Lemma \ref{lem:wheel} indicates that $H(G) \ge 4$. Thus, $G$ contains $C_5$ as an induced subgraph, say by vertices $v_1,\dots,v_5$. The rest of the proof will show that the entropy of $G$ is $11/3$. The first step, culminating in Claim \ref{claim:G(67)} characterises $G[\{v_6, \dots, v_n\}]$.

\begin{claim} \label{claim:one_edge}
$G[\{v_6, \dots, v_n\}]$ has at most one edge.
\end{claim}

\begin{proof}
Let $D = G[\{v_6, \dots, v_n\}]$ have at least two edges. If $D$ has a triangle or a matching of size two, then $H(G) \ge H(C_5) + 2 \ge 4$. Hence all the edges in $D$ are adjacent to one vertex, say $v_6$, and let $v_6v_7, v_6v_8 \in E$. If $N(v_7; G) = \{v_6\}$, then denoting $S = \{v_6\}$, we have $v_7 \in c(S)$ and hence the edge $v_6v_7$ is a matching in $G[c(S), S]$ saturating $S$, which contradicts Lemma \ref{lem:claim}. Therefore, $v_7$ is adjacent to some vertex in $\{v_1, \dots, v_5\}$, say $v_1v_7 \in E$. Then $\{v_2v_3, v_4v_5, v_1v_7, v_6v_8\}$ is a matching in $G$ and $H(G) \ge 4$.
\end{proof}

\begin{claim} \label{claim:exactly_one_edge}
$G[\{v_6, \dots, v_n\}]$ has exactly one edge.
\end{claim}

\begin{proof}
If $G$ has six vertices, then its entropy is either $3$ or $3.5$ by Lemma \ref{lem:wheel}. 

Let $G$ contain at least seven vertices, such that $\{v_6, \dots, v_n\}$ is an independent set. We first remark that for any $i \ge 6$, $|N(v_i; G)| \ge 2$. Indeed, if $v_i$ is only adjacent to say $v_j$ ($1 \le j < i \le n$), then $S = \{v_j\}$ violates Lemma \ref{lem:claim}. Moreover, $N(\{v_6,\dots, v_n\}; G)$ is not contained in a minimum vertex cover of $G[\{v_1, \dots, v_5\}]$, since otherwise $H(G) \le \feedback(G) = \feedback(C_5) = 3$. Hence $N(\{v_6, \dots, v_n\})$ contains three consecutive vertices in the pentagon, say $v_1, v_2, v_3 \in N(\{v_6, \dots, v_n\})$. Say $v_6$ has the largest degree amongst $v_6, \dots, v_n$; we finish the proof by a case analysis on $N(v_6)$.
\begin{enumerate}
	\item \label{it:1} $\{v_1, v_2, v_3\} \subseteq N(v_6)$. Then $v_1v_2v_6$, $v_2v_3v_6$, and all the other edges of the pentagon are cliques. Since $v_7$ has at least two neighbours on the pentagon, $\{v_1, v_3, v_4, v_5\} \cap N(v_7) \ne \emptyset$. We then consider all possible vertices for the intersection. If $v_1v_7 \in E$, then $\{v_2v_3v_6, v_3v_4, v_5v_7\}$ are disjoint cliques and hence $H(G) \ge 4$; the case $v_3v_7 \in E$ is similar. If $v_4v_7 \in E$, then $\{v_2v_3v_6, v_1v_2, v_4v_7\}$ are disjoint cliques and hence $H(G) \ge 4$; the case $v_5v_7 \in E$ is similar. 
	
	\item \label{it:2} $N(v_6) = \{v_1, v_2\}$. Then $v_3$ must be contained in the neighbourhood of a vertex in $\{v_7, \dots, v_n\}$, say $v_3 \in N(v_7)$. Then $\{v_1 v_2 v_6, v_3 v_7, v_4 v_5\}$ are disjoint cliques and hence $H(G) \ge 4$. The case where $N(v_6) = \{v_2, v_3\}$ is similar.
	
	\item \label{it:3} $N(v_6) = \{v_1, v_3\}$. Denoting $S = \{v_1, v_3\}$, we have $c(S) \subseteq \{v_2, v_6\}$ and $G[S, c(S)]$ contains the matching $\{v_1v_2, v_3v_6\}$ saturating $S$, which violates Lemma \ref{lem:claim}.
\end{enumerate}
\end{proof}

\begin{claim}\label{claim:G(67)}
$n=7$ and $G[\{v_6, v_7\}] = K_2$.
\end{claim}

\begin{proof}
We only have to prove that $G$ has no more than seven vertices. Suppose $G$ has eight vertices or more, and that the only edge in $G[\{v_6, \dots, v_n\}]$ is $v_6v_7$. Say $v_1v_8 \in E$, then $\{v_1v_8, v_2v_3, v_4v_5, v_6v_7\}$ is a matching in $G$ and hence $H(G) \ge 4$.
\end{proof}

The second step reduces the search to just six possible graphs. We define the graphs $G_1, \dots, G_6$ displayed in Figure \ref{fig:graphs}.

\begin{claim}
$G \in \{G_1, \dots, G_6\}$.
\end{claim}

\begin{proof}
First, $N(v_6) \cap N(v_7) = \emptyset$. Otherwise, say $v_1v_6, v_1v_7 \in E$, then $\{v_1v_6v_7, v_2v_3, v_4v_5\}$ are disjoint cliques and $H(G) \ge 4$. Suppose $|N(v_6)| \ge |N(v_7)|$, we then have $|N(v_7)| \ge 2$.

\end{proof}
\begin{enumerate}
	\item \label{it:1'} $\{v_1, v_2, v_3\} \subseteq N(v_6)$. Then since $v_2v_7 \notin E$, we can use Case \ref{it:1} of the proof of Claim \ref{claim:exactly_one_edge} to show that $H(G) \ge 4$.
	
	\item \label{it:2'} $N(v_6) = \{v_1, v_2, v_4, v_7\}$. Then $v_3v_7 \in E$ (or $v_5v_7 \in E$, but this is equivalent by symmetry) and $\{v_1v_2v_6, v_3v_7, v_4v_5\}$ are disjoint cliques and $H(G) \ge 4$.
	
	\item \label{it:3'} $N(v_6) = \{v_1, v_2, v_7\}$. Again, if $v_3v_7 \in E$ (or $v_5v_7 \in E$) then $\{v_1v_2v_6, v_3v_7, v_4v_5\}$ are disjoint cliques and $H(G) \ge 4$. Otherwise, $G = G_1$.
	
	\item \label{it:4'} $N(v_6) = \{v_1, v_3, v_7\}$. Then since $v_6 \in N(v_7)$, $v_1, v_3 \notin N(v_7)$, and $2 \le |N(v_7)| \le 3$, we have $G \in \{G_2, G_3, G_4\}$.
	
	\item \label{it:5'} $N(v_6) = \{v_1, v_7\}$. Then $G \in \{G_5, G_6\}$.
\end{enumerate}

The third and final step determines the entropy of $G_1, \dots, G_6$.

\begin{claim}
$H(G_1) = 11/3$ and $H(G_2) = \dots = H(G_6) = 7/2$.
\end{claim}

\begin{proof}
Direct calculations show that the Shannon entropy of $G_1$ is $\ShannonEntropy(G_1) = 11/3$. Conversely, the fractional clique covering number is $10/13$, which is achieved by the following weights:
$$
	w(v_1v_2v_6) = w(v_3v_4) = w(v_4v_5) = w(v_4v_7) = 1/3, \quad w(v_1v_5) = w(v_2v_3) = w(v_6v_7) = 2/3.
$$
Therefore, $11/3 = n - \fractionalCliqueCover(G_1) \le H(G_1) \le \ShannonEntropy(G_1) = 11/3$. For $G_i$, $i \ge 2$, direct calculations show that $\ShannonEntropy(G_i) = 7/2$. Conversely, $H(G_i) \ge H(C_5) + 1 = 7/2$.
\end{proof}
\end{proof}

\begin{figure}
\centering
	\subfloat[$G_1$]
	{\begin{tikzpicture}
	\tikzstyle{every node}=[draw,shape=circle];

	\node (1) at (90-72*1:2) {1};
	\node (2) at (90-72*2:2) {2};
	\node (3) at (90-72*3:2) {3};
	\node (4) at (90-72*4:2) {4};
	\node (5) at (90-72*5:2) {5};	
	
    \draw (1) -- (2);
    \draw (2) -- (3);
    \draw (3) -- (4);
    \draw (4) -- (5);
    \draw (5) -- (1);

	\node (6) at (0,1) {6};
	\node (7) at (0,-1) {7};
	
	\draw (6) -- (7);
	
	\draw (6) -- (1);
	\draw (6) -- (2);
	\draw (7) -- (4);
	\end{tikzpicture}} \hspace{1cm}
	\subfloat[$G_2$]
	{\begin{tikzpicture}
	\tikzstyle{every node}=[draw,shape=circle];

	\node (1) at (90-72*1:2) {1};
	\node (2) at (90-72*2:2) {2};
	\node (3) at (90-72*3:2) {3};
	\node (4) at (90-72*4:2) {4};
	\node (5) at (90-72*5:2) {5};	
	
    \draw (1) -- (2);
    \draw (2) -- (3);
    \draw (3) -- (4);
    \draw (4) -- (5);
    \draw (5) -- (1);

	\node (6) at (0,1) {6};
	\node (7) at (0,-1) {7};
	
	\draw (6) -- (7);
	
	\draw (6) -- (1);
	\draw (6) -- (3);
	\draw (7) -- (2);
	\end{tikzpicture}} \hspace{1cm}
	\subfloat[$G_3$]
	{\begin{tikzpicture}
	\tikzstyle{every node}=[draw,shape=circle];

	\node (1) at (90-72*1:2) {1};
	\node (2) at (90-72*2:2) {2};
	\node (3) at (90-72*3:2) {3};
	\node (4) at (90-72*4:2) {4};
	\node (5) at (90-72*5:2) {5};	
	
    \draw (1) -- (2);
    \draw (2) -- (3);
    \draw (3) -- (4);
    \draw (4) -- (5);
    \draw (5) -- (1);

	\node (6) at (0,1) {6};
	\node (7) at (0,-1) {7};
	
	\draw (6) -- (7);
	
	\draw (6) -- (1);
	\draw (6) -- (3);
	\draw (7) -- (4);
	\end{tikzpicture}} \hspace{1cm}
	\subfloat[$G_4$]
	{\begin{tikzpicture}
	\tikzstyle{every node}=[draw,shape=circle];

	\node (1) at (90-72*1:2) {1};
	\node (2) at (90-72*2:2) {2};
	\node (3) at (90-72*3:2) {3};
	\node (4) at (90-72*4:2) {4};
	\node (5) at (90-72*5:2) {5};	
	
    \draw (1) -- (2);
    \draw (2) -- (3);
    \draw (3) -- (4);
    \draw (4) -- (5);
    \draw (5) -- (1);

	\node (6) at (0,1) {6};
	\node (7) at (0,-1) {7};
	
	\draw (6) -- (7);
	
	\draw (6) -- (1);
	\draw (6) -- (3);
	\draw (7) -- (2);
	\draw (7) -- (4); 
	\end{tikzpicture}} \hspace{1cm}
	\subfloat[$G_5$]
	{\begin{tikzpicture}
	\tikzstyle{every node}=[draw,shape=circle];

	\node (1) at (90-72*1:2) {1};
	\node (2) at (90-72*2:2) {2};
	\node (3) at (90-72*3:2) {3};
	\node (4) at (90-72*4:2) {4};
	\node (5) at (90-72*5:2) {5};	
	
    \draw (1) -- (2);
    \draw (2) -- (3);
    \draw (3) -- (4);
    \draw (4) -- (5);
    \draw (5) -- (1);

	\node (6) at (0,1) {6};
	\node (7) at (0,-1) {7};
	
	\draw (6) -- (7);
	
	\draw (6) -- (1);
	\draw (7) -- (2);
	\end{tikzpicture}} \hspace{1cm}
	\subfloat[$G_6$]
	{\begin{tikzpicture}
	\tikzstyle{every node}=[draw,shape=circle];

	\node (1) at (90-72*1:2) {1};
	\node (2) at (90-72*2:2) {2};
	\node (3) at (90-72*3:2) {3};
	\node (4) at (90-72*4:2) {4};
	\node (5) at (90-72*5:2) {5};	
	
    \draw (1) -- (2);
    \draw (2) -- (3);
    \draw (3) -- (4);
    \draw (4) -- (5);
    \draw (5) -- (1);

	\node (6) at (0,1) {6};
	\node (7) at (0,-1) {7};
	
	\draw (6) -- (7);
	
	\draw (6) -- (1);
	\draw (7) -- (3);
	\end{tikzpicture}}
	\caption{The graphs $G_1, \dots, G_6$.} \label{fig:graphs}
\end{figure}
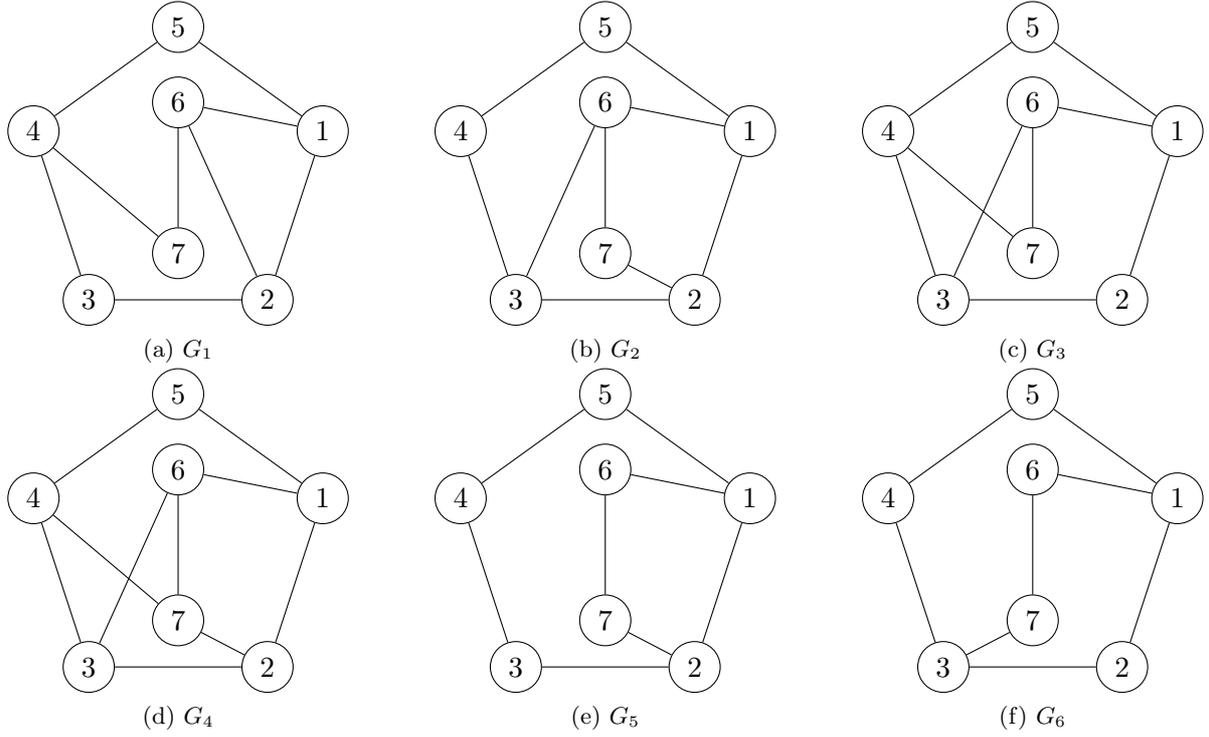

\section{Acknowledgment}
This work was supported by the Engineering and Physical Sciences Research Council [grant number EP/K033956/1].


\providecommand{\bysame}{\leavevmode\hbox to3em{\hrulefill}\thinspace}
\providecommand{\MR}{\relax\ifhmode\unskip\space\fi MR }
\providecommand{\MRhref}[2]{%
  \href{http://www.ams.org/mathscinet-getitem?mr=#1}{#2}
}
\providecommand{\href}[2]{#2}

\end{document}